\documentclass[runningheads]{llncs}
\usepackage[T1]{fontenc}
\usepackage{mathtools}

\usepackage{graphicx}
\spnewtheorem*{definition*}{Definition}{\bfseries}{\itshape}
\PassOptionsToPackage{hyphens}{url}\usepackage{hyperref}
\usepackage{cleveref}
\usepackage{color}

\usepackage{csquotes}
\usepackage{tikz}
\usepackage{pgfplots}
\usepackage{caption}
\usepackage{subcaption}
\pgfplotsset{compat=1.18}
\usepgfplotslibrary{fillbetween}
\crefname{equation}{}{}
\usepackage{cite}

\begin{document}
\title{Retention Induced Biases in a Recommendation System with Heterogeneous Users}
\titlerunning{Retention Bias with Heterogeneous Users}
%
\author{Shichao Ma\orcidID{0000-0001-6143-103X}}
\authorrunning{S. Ma}
%
\institute{Yelp Inc., San Francisco CA 94105, USA \\
\email{shichaom@yelp.com}}
\maketitle              
\begin{abstract}
I examine a conceptual model of a recommendation system (RS) with user inflow and churn dynamics. 
When inflow and churn balance out, the user distribution reaches a steady state. Changing the recommendation algorithm alters the steady state and creates a transition period. During this period, the RS behaves differently from its new steady state. 
In particular, A/B experiment metrics obtained in transition periods are biased indicators of the RS’s long-term performance. 
Scholars and practitioners, however, often conduct A/B tests shortly after introducing new algorithms to validate their effectiveness. 
This A/B experiment paradigm, widely regarded as the gold standard for assessing RS improvements, may consequently yield false conclusions. 
I also briefly touch on the data bias caused by the user retention dynamics.

\end{abstract}
\section{Introduction}

Running a recommendation system (RS) is rarely a one-shot game. For example, YouTube was founded in 2005 and has been serving users for nearly two decades.\footnote{\url{https://en.wikipedia.org/wiki/History_of_YouTube}}
In any long running RS, users join and leave over time. To be successful in the long run, it is critical for developers to constantly enhance the RS, attract new users and retain existing ones. 

A key improvement of great theoretical value is improving the quality of the recommendation. Throughout this paper, recommendation quality refers to the general attractiveness of the recommended items and the general efficacy of the recommendation algorithm. It represents the core user experience that an RS can offer. 
Users may lose interest and leave the platform if they find the recommended items irrelevant.
Conversely, high recommendation quality may encourage more users to stay, which leads to the long term success of the RS.

However, the mainstream RS literature does not explicitly consider the above dynamics. Most studies assume the user is a fixed set or has a fixed distribution and focus on recommendation algorithms that optimize some transient measures of recommendation quality given the fixed user distribution \cite{Salakhutdinov2008ProbabilisticFactorization,Rendle2010FactorizationMachines,Cheng2017LearningItems,Hidasi2016Session-basedNetworks,He2017NeuralAnalytics,Shi2010List-wiseFiltering,Weimer2008CoFiRANKRanking,Chen2020DeepRank:Recommendation,Liu2010PersonalizedBehavior,Huang2019OnlineDynamics,Yi2014BeyondPersonalization,Diaz-Aviles2014PredictingRanking,Guo2012BeyondBehavior}.

This literature leaves an important question unanswered: Does an RS that explicitly optimizes recommendation quality also implicitly optimize user retention, maximize user population and bring about the system's long-term success? If yes, then the mainstream approach, discovering new learning algorithms that optimize some forms of recommendation quality, is justified. 
If not, then practitioners may not benefit from the most advanced learning algorithm.

Some notable exceptions in the reinforcement learning (RL) literature address the optimization of user retention without answering the above research question \cite{Zou2019ReinforcementSystems,Cai2023ReinforcingSystem,Xue2022ResAct:Actor,Wu2017ReturningSystems}.
However, these approaches do not discuss why explicitly optimizing retention is necessary.
As \cite{Cai2023ReinforcingSystem} admit, ``the relationship between retention reward and the intermediate feedback is not clear.'' Moreover, even though the RL literature acknowledges the importance of churn, it still ignores \textit{user inflow}. Considering user inflow is obviously more faithful to the reality when we think about RS, but does it have any nontrivial implications that justify the extra complexity?

This paper differs from the existing literature by modeling the user as a flow and analyzing the system's equilibrium behavior. Conceptually, I demonstrate that the introduction of user inflow has profound and counterintuitive implications on how we think about improving recommendation algorithms. User inflow and churn are two competing forces that affect user population. When these two forces strike a balance, the system reaches a steady state in which the user distribution remains constant. 
However, the RS may take a long time to converge to this state. During this time, the system may differ significantly from its steady-state behavior.

Practically, I show that A/B testing after changing the recommendation algorithm may lead to biased test results. 
Changes in recommendation quality may alter user churn and disrupt the existing steady state. 
Because developers usually A/B test their systems shortly after making changes, they measure test metrics during the transition periods.
These transitory metrics do not necessarily agree with their steady-state counterparts.
More critically, changes that harm the system in the long run may seem beneficial in the A/B test.

The user inflow and churn dynamics also hints a new form of data bias in developing RS algorithms. In their excellent survey, \cite{Chen2023BiasDirections} define data bias as the deviation from the underlying distribution.
In this sense, the departure from the long-term steady-state distribution is also a data bias. Paradoxically, when a developer tries a new model architecture, she must train the new model using existing data, which is generally biased from the new architecture's steady state. What are the implications of this bias? How can it be mitigated? Given this paper's theoretical focus, I leave these questions to future research.

\subsubsection{Related Work}

Mainstream recommendation models generally can be framed as supervised learning problems where the ground-truth training labels are explicit user-item ratings or implicit 
feedback signals \cite{Koren2009MatrixSystems}. The goal is to accurately predict the training labels. Examples of general-purpose approaches include probabilistic matrix factorization \cite{Salakhutdinov2008ProbabilisticFactorization}, factorization machine \cite{Rendle2010FactorizationMachines}, and neural factorization machine \cite{He2017NeuralAnalytics}. A slight reformulation of the problem aims to predict ranking as accurately as possible. Examples include \cite{Shi2010List-wiseFiltering,Weimer2008CoFiRANKRanking,Chen2020DeepRank:Recommendation}. For implicit feedback signals, scholars have tried click \cite{Liu2010PersonalizedBehavior}, purchase \cite{Huang2019OnlineDynamics}, dwell time \cite{Yi2014BeyondPersonalization}, user engagement \cite{Diaz-Aviles2014PredictingRanking}, cursor movement \cite{Guo2012BeyondBehavior}, etc. Within this literature, training labels are typically viewed as instantaneous signals directly linked to users' current behavior. Furthermore, this body of work often assumes either implicitly or explicitly that users form a fixed set or follow a fixed distribution.
For example, \cite{Salakhutdinov2008ProbabilisticFactorization} write: ``Suppose we have $M$ movies, $N$ users...''; \cite{Weimer2008CoFiRANKRanking} write: ``Assume that we have $m$ items and $u$ users''; \cite{Cheng2017LearningItems} write: ``Let $U$ and $I$ be the sets of users and items respectively, where the cardinalities $|U| = m$ and $|I| = n$.'' 

A major exception that relaxes the assumption of the fixed user distribution is the RL literature. \cite{Cai2023ReinforcingSystem} develop an algorithm that directly minimizes the cumulative discounted return time. \cite{Wu2017ReturningSystems} propose a bandit-based solution to balance immediate feedback, future clicks, and exploration. \cite{Xue2022ResAct:Actor} describe a residual learning approach to improve upon the deployed production policy. \cite{Zou2019ReinforcementSystems} devise a method to optimize a composite inter-temporal user engagement metric. In this literature, users can churn explicitly or implicitly through performing null actions. Nevertheless, this literature does not consider user inflow as an integral part of their modeling.

Many articles that explore new RS algorithms are susceptible to the pitfall identified in this paper. \cite{Cheng2017LearningItems,Hidasi2016Session-basedNetworks,Liang2016ModelingRecommendation,Rendle2010FactorizationMachines,Salakhutdinov2008ProbabilisticFactorization,Xue2022ResAct:Actor} show performance gains of their respective algorithms only on some fixed benchmark datasets. \cite{Wu2017ReturningSystems} retroactively compare their solution against a battery of baseline bandit algorithms on Yahoo news recommendation logs. \cite{Zou2019ReinforcementSystems} evaluate their method on two days' data after their training time window. \cite{Guo2019PAL:Systems} measure the first-week click-through rate and conversion rate after turning on the experiment as their A/B test metrics.
\cite{Cai2023ReinforcingSystem} report live test metrics from day 0 to day 150 to claim superiority over their status quo. Even though their A/B test is much longer than others, it is unclear if 150 days are enough for their system to converge to the new steady state. \cite{Amatriain2015RecommenderStudy} describe their offline-online testing procedure for RS innovations in Netflix. They write:
\begin{displayquote}
Once offline testing has validated a hypothesis, we are ready to design and launch the online AB test that will demonstrate if an algorithmic change is an improvement from the perspective of user behavior. If it does, we will be ready to deploy the algorithmic improvement to the whole user-base.
\end{displayquote}
However, according to this paper, it might be unwise to launch an online A/B test ``once offline testing has validated a hypothesis'' because Netflix's monthly billing model can cause their system taking a long time to converge.

The existing literature on RS bias largely focuses on counterfactual biases like feedback loop \cite{Mansoury2020FeedbackSystems}, selection bias \cite{Wang2016LearningSearch}, exposure bias \cite{Liang2016ModelingRecommendation}, position bias \cite{Wang2018PositionSearch,Guo2019PAL:Systems}, and so on. A recent comprehensive literature survey on biases in RS can be found at \cite{Chen2023BiasDirections}. However, to the best of my knowledge, the bias caused by user inflow and churn dynamics had not been identified in the literature.

This paper also relates to the literature on user satisfaction with RS. \cite{Kim2021CustomerApproaches} note that accuracy and diversity can influence user satisfaction in RS. \cite{Nguyen2018UserSystems} show that user personality contributes to their satisfaction with RS. Based on these empirical findings, I incorporate both recommendation quality and unobserved factors into the model setup discussed below, assuming they impact user retention.

\section{Model Setup and the Status Quo}

I consider a world where time is discrete and infinite. At the beginning of each period, a continuum of new users with unit mass arrive and start to consume the contents that an RS recommends. For simplicity, I assume that the new user distribution can be parameterized by two variables, $x$ and $e$, and I use $F(x, e)$ to denote this distribution function. I assume $x$ is an observable feature to the RS. The machine learning model behind the RS may use $x$ to compute its predictions. In practice, $x$ can represent self-reported user demographics, signals collected from past user behaviors, or data obtained from third-party vendors about users' general web activities. The variable $e$, however, is hidden from the RS and includes private user information relevant to RS consumption, such as personality, job, wealth, education level, and offline activities.

To further simplify the setup, I assume that $x$ and $e$ can each take only two values: $3/4$ or $1/4$. This creates four user types based on their $(x,e)$ combinations.
I also make the following symmetric distributional assumption: $F(x=3/4, e=3/4) = F(x=1/4, e=1/4) = \alpha$ and $F(x=3/4, e=1/4) = F(x=1/4, e=3/4) = 1/2 - \alpha$, where $\alpha \in (0, 1/2)$ is an exogenous parameter. 

At the end of each period, each existing user decides whether to churn or not. Churn means the user stops using the RS. Once a user churns, they never return.\footnote{In the real world, churned users may change their mind and come back, but in this model, they are treated as new users.}
For the purpose of this paper, I abstract away the specific contents recommended by the RS. Instead, I assume that user's churn decision can be influenced by a statistic of items recommended during the period.\footnote{Thus, recommendation quality is a short-term signal by assumption, aligning with the mainstream literature that uses short-term signals to train RS models.} I interpret this statistic recommendation quality and denote it as $q$, where $q \in (0, 1)$. 
In the real world, the recommendation quality statistic may have many names. In an ad recommendation system, it may be the user-level click-through rate. In a movie recommendation system, it may be the view rate or completion rate. In an online shopping system, it may be the conversion rate.

As shown by empirical studies \cite{Erdem2023OnPerformances}, different user segments may have different performance. Because $x$ is an observable but $e$ is not, $q$ must be a function mapping from the feature set, $x \in \{1/4, 3/4\}$, to the recommendation quality statistic. This function summarizes the intricacies of a recommendation algorithm as a simple mapping.
As the system only sees two segments of users, I call $x=3/4$ the high segment and $x=1/4$ low segment. 

Formally, I assume that a user who has not churned has a probability of $(1 - q) (1 - e)$ to churn at the end of the period. The term $(1 - q)$ indicates that as the recommendation quality increases, the likelihood of user churn decreases. For $1-e$, it is known as the frailty term in the survival analysis literature \cite{Balan2020AModels}. The frailty term denotes user's unobserved heterogeneity that cannot be explained by observed covariates \cite{Nguyen2018UserSystems}. Higher $e$ implies a healthier user, who is less likely to churn for reasons unobserved to the RS, and vice versa. A user who does not churn at the end of a period is called a retained user for that period.

The user distribution in the system may change over time due to user inflow and churn. It may also reach a stable point where user inflow and churn are balanced. This point is called the steady state, and is defined as follows.
\begin{definition*}[Steady State]
The steady state of the system is an equilibrium point where user inflow fully offsets churned users at each period for each type. Formally, $m_\infty(x, e) \cdot \Pr(\mathrm{churn}|x, e) = F(x, e)$ for all $(x, e)$, where $m_\infty(x, e)$ is the population of type $(x, e)$ users retained in the steady state.
\end{definition*}
All steady-state quantities are subscripted by $\infty$. The steady-state user distribution can be obtained by $m_\infty(x, e) = F(x, e) / \Pr(\mathrm{churn}|x, e)$. 

\subsubsection{Success Metrics} 

In our current setup, two metrics that measure the success of RS are naturally defined: \textit{retained user population} and \textit{average recommendation quality (ARQ)}. 
To assess the long-term success of the system, I focus on their steady-state values.

\paragraph{Steady State User Population}
For many tech companies, getting more users and grow is a constant theme. Thus, I consider the user population retained in the steady-state, denoted as $m_\infty$, the core metric to measure the RS success. Its real world counterpart is per-period (e.g., daily or monthly) active or paying user, which is considered a top-line metric by many real world RSs \cite{Facebook2015,Murariu2023}. In our current setup, $m_\infty = \sum_{x}\sum_{e} m_\infty(x, e)$.

\paragraph{Average Recommendation Quality (ARQ)}
Developers may routinely track aggregated recommendation quality metrics like click-through rate, completion rate, or average rating. Arguably, these metrics are the most direct measures of an RS’s performance. Therefore, developers may naturally focus on enhancing them. For example, a developer working on a video recommendation system might aim to boost the watch rate. Similarly, a developer of an ad recommendation system might consider the average click-through rate as a key metric. Using $\overline{q}$ to denote ARQ, we have $\overline{q}_\infty = \left(\sum_x q(x) (\sum_e m(x, e))\right) / (\sum_{x}\sum_{e} m(x, e))$.

\subsubsection{Status Quo} 
Suppose that the RS has been using an existing recommendation algorithm, called status quo, for a long time and has reached its steady state. I use $*$ to superscript quantities of the status quo. As all metrics in the status quo have reached their steady state values, I omit their $\infty$ subscripts to simplify notation. Normalizing $q^*(x)=x$ to be the recommendation quality function in the status quo, we can easily obtain $m^*(3/4, 3/4) = 16\alpha$, $m^*(3/4, 1/4) = m^*(1/4, 3/4) = 8(1-2\alpha)/3$, $m^*(1/4, 1/4) = 16\alpha/9$, $m^* = 16(4\alpha + 3)/9$, and $\overline{q}^* = (8\alpha+3)/(8\alpha+6)$.

\section{Treatment Algorithm and A/B Experimentation}

In this section, I consider a stereotypical situation where a developer aims to improve the recommendation algorithm and evaluates a treatment algorithm via an A/B test. An A/B test is a side-by-side comparative method that assesses two algorithms to determine which performs better. I use $T$ to superscript quantities of the treatment. Given that $q^*(x)$ encapsulates the status quo recommendation algorithm, any enhancement to the current algorithm boils down to modifying this function. To simplify the notation, I define $q_3 \equiv q^T(3/4)$ and $q_1 \equiv q^T(1/4)$.

The steady-state behavior of the new algorithm can be studied immediately. 

\begin{proposition}\label{prop:long-term-no-conflict}
We have that $m^T_\infty \geq (<)\ m^*$ if and only if $\overline{q}^T_\infty \geq (<)\ \overline{q}^*$.
\end{proposition}
All proofs can be found in the appendix.
This proposition justifies the bedrock of the mainstream RS literature: optimizing for recommendation quality and gaining users are perfectly aligned in the long run. From now on, I use \textit{steady-state performance} as a collective term to denote either of the steady state metrics.

However, \Cref{prop:long-term-no-conflict} does not imply that A/B testing the treatment algorithm will always produce the correct outcome.

\subsubsection{A/B Experimentation}

Concretely, at the onset of some period, the treatment algorithm is implemented on a random subset of users. To simplify the discussion of experimental results, I assume that this subset is also a continuum that inherits all statistical properties of the existing user distribution. 
The A/B experimentation lasts for one period, to mirror the real-world scenario where A/B tests are typically much shorter than the longevity of the system. 
The developer collects data at the end of the period and makes a verdict between the status quo and the treatment.

I consider the following two metrics to measure the experiment success: ARQ and churn rate. Notice metrics observed in the experiment, subscripted with $E$, are generally different from their steady state counterparts. In our current setup,
\begin{equation}\label{eq:rq_def}
    \overline{q}^T_E = \frac{1}{m^*}
    \sum_x \left( q^T(x) \left( \sum_e m^*(x, e) \right) \right).
\end{equation}

\paragraph{Churn Rate} 
Churn rate is the percentage of users who churn out of the total population of existing users during the experiment period. In some situations, the developer may directly monitor user churn. This could be because the developer understands that her work influences churn and does not want to lose users. It may also be that the developer cannot directly observe the recommendation quality and uses churn rate as a substitute. Using $\lambda$ to denote churn rate, we have $\lambda^* = 1/m^*$ and 
\begin{equation}\label{eq:cr_def}    
    \lambda^T_E = \frac{1}{m^*}
    \sum_x \left( (1-q^T(x)) \left(\sum_e m^*\left(x, e\right) (1-e)\right) \right).
\end{equation}
It is important to note that both \cref{eq:rq_def,eq:cr_def} use the user distribution of the status quo. This is because the user distribution has not changed by the treatment algorithm yet when it is first introduced.

The following lemma suggests that ARQ and churn rate may provide contradicting signals to the experiment's success, despite optimizing for recommendation quality and gaining users are perfectly aligned in the long run.

\begin{lemma}\label{lemma:q-churn-not-same}
If 
\begin{equation}\label{eq:lemma1}
    \left((4\alpha-3)q_1 + 8\alpha + 3\right)/(12\alpha+3) < q_3 < (5-2q_1) / 6
\end{equation}
and either $q_1 < 1/4 < \alpha$ or $\alpha < 1/4 < q_1$ holds, then $\overline{q}^T_E > \overline{q}^*$ and $\lambda^T_E > \lambda^*$.
\end{lemma}
Practitioners may notice that quality metrics and churn do not always go hand in hand in reality. \Cref{lemma:q-churn-not-same} provides a theoretical explanation.
The driving force here is the frailty term $1-e$. Because of it, the marginal contributions of $\Delta q_1$ and $\Delta q_3$ to the hazard rate vary among user types. However, it does not affect the marginal contributions to ARQ.
Thus, when $q_1$ and $q_3$ change by the ``correct'' amount, one metric may appear to improve while the other may seem to regress.

The following corollary summarizes the conditions under which the experiment is deemed a ``success'' by either metric. When it holds, it’s reasonable to assume that the developer would favor the treatment over the status quo.
\begin{corollary}\label{cor:ab-win}
Suppose $q_1$ and $q_3$ satisfy either of the following conditions:
\begin{enumerate}
    \item If $\alpha = 1/4$, $6q_3+2q_1>5$.
    \item If $\alpha < 1/4$, $(q_1, q_3) \in \{(q_1, q_3)|q_1 \geq 1/4, 6q_3+2q_1>5\} \cup \{(q_1, q_3)|q_1 < 1/4, (12\alpha+3)q_3+(3-4\alpha)q_1>8\alpha+3\}$.
    \item If $\alpha > 1/4$, $(q_1, q_3) \in \{(q_1, q_3)|q_1 \leq 1/4, 6q_3+2q_1>5\} \cup \{(q_1, q_3)|q_1 > 1/4, (12\alpha+3)q_3+(3-4\alpha)q_1>8\alpha+3\}$.
\end{enumerate}
The developer observes $\overline{q}^T_E > \overline{q}^*$ and $\lambda^T_E < \lambda^*$ in her experiment.
\end{corollary}
In A/B experiments, success metrics, regardless of their definitions, are generally computed based on existing users. 
A change to the recommendation algorithm may impact different user segments in varying ways, possibly altering the steady state.
The system moves to its new steady state through user inflow and churn.
Given that churn takes time to manifest, convergence to the new steady state can be a lengthy process. Thus, metrics based on existing users may not always provide insights into the system’s future behavior.
In fact, relying on such metrics could be misleading for the long-term success of the RS.
The rest of the paper demonstrates this insight in detail.

\begin{lemma}[ARQ Improvement]\label{lemma:q-up-mass-down}
When $q_1$ and $q_3$ satisfy
\begin{equation}\label{eq:q-up}
    (12\alpha+3)q_3 + (3-4\alpha)q_1 > 8\alpha+3
\end{equation}
and
\begin{equation}\label{eq:mass-down}
    (4\alpha+1)/(1-q_3) + (3-4\alpha)/(1-q_1) < 8(4\alpha+3)/3,
\end{equation}
the developer observes $\overline{q}^T_E > \overline{q}^*$ in her experiment. However, $\overline{q}^T_\infty < \overline{q}^*$.
\end{lemma}
Remarkably, even though the developer observes an improvement in ARQ in the experiment, the new algorithm may decrease this quality in the long run. 
To further illustrate this point, consider a simulation where the developer applies a treatment algorithm with $q_1 = 7/16, q_3=45/64$ to all users in a system with $\alpha=1/4$. The simulated ARQ path, plotted in the left panel of \cref{fig:arq_simulation}, shows that ARQ hikes following the introduction of the treatment. However, it gradually converges to its new steady state value, which is significantly lower than the status quo value.

\begin{figure}[t]
\centering
\begin{subfigure}[t]{0.49\textwidth}
\centering
\begin{tikzpicture}
\begin{axis}[
    axis lines = left,
    xmin=0, xmax=42, 
    ymin=0.61, ymax=0.64,
    scale=0.6, 
    xlabel={Periods after Treatment's Introduction},
    ylabel={ARQ},
    extra x ticks={1},
    transform shape,
    clip=false,
    label style={font=\footnotesize},
    tick label style={font=\footnotesize},
    scale only axis=true,
    width=1.3\textwidth,
    height=0.5\textwidth,
]
\addplot[color=black, mark=none, thick] table[x index=0,y index=1,col sep=comma] {arq.csv};
\draw[dotted] (0,0.625) -- (1,0.63671875); 
\filldraw[black] (0,0.625) circle (1pt) node[anchor=west]{};
\filldraw[black] (1,0.63671875) circle (1pt) node[anchor=west]{};
\end{axis}
\end{tikzpicture}
\end{subfigure}
\begin{subfigure}[t]{0.49\textwidth}
\centering
\begin{tikzpicture}
\begin{axis}[
    axis lines = left,
    xmin=0, xmax=42, 
    ymin=6.85, ymax=7.17,
    scale=0.6, 
    xlabel={Periods after Treatment's Introduction},
    ylabel={User Population},
    ytick={6.85,6.95,7.05,7.15},
    extra x ticks={2},
    transform shape,
    label style={font=\footnotesize},
    tick label style={font=\footnotesize},
    scale only axis=true,
    width=1.3\textwidth,
    height=0.5\textwidth,
]
\addplot[color=black, mark=none, thick] table[x index=0,y index=1,col sep=comma] {data.csv};
\draw[dotted] (2,6.85) -- (2,7.152370876736111);
\end{axis}
\end{tikzpicture}
\end{subfigure}
\caption{ARQ and Population Dynamics after Introducing a Treatment Algorithm}
\label{fig:arq_simulation}
\end{figure}
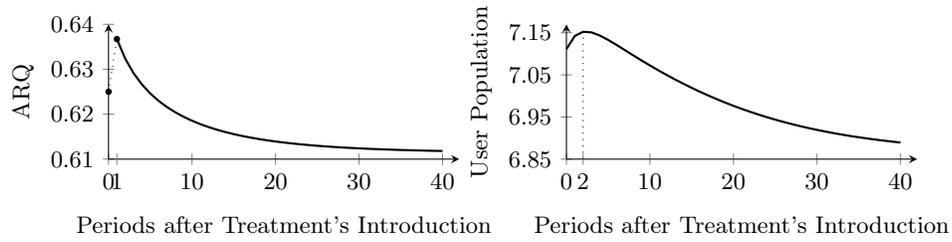

\begin{lemma}[Churn Reduction]\label{lemma:churn-down-mass-down}
When $q_1$ and $q_3$ satisfy
\begin{equation}\label{eq:churn-down}
    6q_3 + 2q_1 > 5
\end{equation}
and \cref{eq:mass-down}, the developer observes $\overline{\lambda}^T_E < \overline{\lambda}^*$ in her experiment. However, $m^T_\infty < m^*$.
\end{lemma}
That is, the developer may observe that more users are retained in the A/B test even if her new algorithm hurts the RS's ability to retain users in the long run. 

The right panel of \cref{fig:arq_simulation} plots the dynamics of the retained user population in the above simulation. It shows that user population increases for two periods after the treatment's introduction and then gradually declines and converges to the new steady state value. However, due to the initial increase, it remains above the status quo level until the seventh period.

Why can the RS retain more users to begin with? In the simulation, $q_1$ rises by 75\%. This significant increase causes the low segment's new steady state to be considerably different from the status quo. As a result, this segment experiences a substantial user growth in the early periods. However, as the low segment's population converges, its growth slows down.
On the other hand, $q_3$ decreases by only 6.25\%. This is a relatively minor impact, which takes longer to manifest. Nevertheless, the cumulative effect of this negative impact eventually surpasses the positive impact from the low segment since the high segment has a larger population to begin with. Consequently, we observe a short-term increase in user population, even though there is a permanent decline in the long run.

Combining \Cref{cor:ab-win}, \Cref{lemma:churn-down-mass-down,lemma:q-up-mass-down}, we reach the core argument of this paper.

\begin{proposition}[Unequivocal Improvements]\label{prop:main_result}
Suppose either set of the following conditions holds:
\begin{enumerate}
    \item If $\alpha > 1/4$, $q_1$ and $q_3$ satisfy \cref{eq:q-up,eq:mass-down}.
    \item If $\alpha \leq 1/4$, $q_1$ and $q_3$ satisfy \cref{eq:churn-down,eq:mass-down}.
\end{enumerate}
The developer observes $\overline{q}^T_E > \overline{q}^*$ and $\overline{\lambda}^T_E < \overline{\lambda}^*$. However, $\overline{q}^T_\infty < \overline{q}^*$ and $m^T_\infty < m^*$.
\end{proposition}
In other words, when $q_1$ and $q_3$ move by ``right'' amounts, the A/B experiment is considered a ``success'' by either metric but the treatment permanently impairs the RS performance in the long run.

\Cref{fig:prop1} plots the $(q_1, q_3)$ pairs that satisfy \Cref{prop:main_result}, which may shed some light on the underlying mechanism. According to the figure, in order for the steady state metrics to stay constant, $q_1$ and $q_3$ must move along a concave curve (the solid lines). Conversely, to maintain either experimental metric neutral, $q_1$ and $q_3$ need only follow a straight line (the dotted or dashed lines). For all $\alpha$, there exists a section of the solid curve that lies above both straight lines.

\begin{figure}[t]
\begin{subfigure}[t]{0.32\textwidth}
\centering
\begin{tikzpicture}
\begin{axis}[
    axis lines = left,
    xlabel = \(q_1\),
    ylabel = {\(q_3\)},
    xmin=0, xmax=0.8, 
    ymin=0.5, ymax=0.9,
    scale=.4, 
    transform shape,
    label style={font=\footnotesize},
    tick label style={font=\footnotesize},
]
\addplot[
    domain=0.25:0.625, 
    thick,
    samples=30, 
    name path=A,
]{1 - 2 / (32/3 - 2 / (1-x))};
\addplot[
    domain=0:0.25, 
    thick,
    samples=30, 
]{1 - 2 / (32/3 - 2 / (1-x))};
\addplot[
    domain=0.625:0.75, 
    thick,
    samples=30, 
]{1 - 2 / (32/3 - 2 / (1-x))};
\addplot[
    domain=0.25:0.625,
    thick,
    densely dashed,
    name path=B,
]{(5 - 2 * x) / 6};
\addplot[
    domain=0:0.25, 
    thick,
    densely dashed,
]{(5 - 2 * x) / 6};
\addplot[
    domain=0.625:0.8,
    thick,
    densely dashed,
]{(5 - 2 * x) / 6};
\addplot[black!30] fill between[of=A and B];
\end{axis}
\end{tikzpicture}
\caption{$\alpha=1/4$}
\end{subfigure}
\begin{subfigure}[t]{0.32\textwidth}
\centering
\begin{tikzpicture}
\begin{axis}[
    axis lines = left,
    xlabel = \(q_1\),
    ylabel = {\(q_3\)},
    xmin=0, xmax=0.8, 
    ymin=0.5, ymax=0.9,
    scale=.4, 
    transform shape,
    label style={font=\footnotesize},
    tick label style={font=\footnotesize},
]
\addplot[
    domain=0.25:(23/44), 
    samples=30, 
    thick,
    name path=A,
]{1 - 5 / (88/3 - 7 / (1-x))};
\addplot[
    domain=0:0.25, 
    thick,
    samples=30, 
]{1 - 5 / (88/3 - 7 / (1-x))};
\addplot[
    domain=(23/44):0.75, 
    thick,
    samples=30, 
]{1 - 5 / (88/3 - 7 / (1-x))};
\addplot[
    domain=0.25:(23/44),
    thick,
    densely dashed,
    name path=B,
]{(5 - 2 * x) / 6};
\addplot[
    domain=0:0.25, 
    thick,
    densely dashed,
]{(5 - 2 * x) / 6};
\addplot[
    domain=(23/44):0.8,
    thick,
    densely dashed,
]{(5 - 2 * x) / 6};
\addplot[
    domain=0:0.8, 
    thick,
    densely dotted,
]{(13 - 7 * x) / 15};
\addplot[black!30] fill between[of=A and B];
\end{axis}
\end{tikzpicture}
\caption{$\alpha=1/6$}
\end{subfigure}
\begin{subfigure}[t]{0.32\textwidth}
\centering
\begin{tikzpicture}
\begin{axis}[
    axis lines = left,
    xlabel = \(q_1\),
    ylabel = {\(q_3\)},
    xmin=0, xmax=0.8, 
    ymin=0.5, ymax=0.9,
    scale=.4, 
    transform shape,
    label style={font=\footnotesize},
    tick label style={font=\footnotesize},
]
\addplot[
    domain=0.25:(17/26), 
    samples=30, 
    thick,
    name path=A,
]{1 - 7 / (104/3 - 5 / (1-x))};
\addplot[
    domain=0:0.25, 
    thick,
    samples=30, 
]{1 - 7 / (104/3 - 5 / (1-x))};
\addplot[
    domain=(17/26):0.77, 
    thick,
    samples=30, 
]{1 - 7 / (104/3 - 5 / (1-x))};

\addplot[
    domain=0:0.8, 
    thick,
    densely dashed,
]{(5 - 2 * x) / 6};
\addplot[
    domain=0.25:(17/26),
    densely dotted,
    thick,
    name path=C,
]{(17 - 5 * x) / 21};
\addplot[
    domain=0:0.25, 
    thick,
    densely dotted,
]{(17 - 5 * x) / 21};
\addplot[
    domain=(17/26):0.8,
    thick,
    densely dotted,
]{(17 - 5 * x) / 21};
\addplot[black!30] fill between[of=A and C];
\end{axis}
\end{tikzpicture}
\caption{$\alpha=1/3$}
\end{subfigure}
\begin{subfigure}[t]{\textwidth}
\centering
\begin{tikzpicture}
\begin{axis}[
    axis lines = left,
    xmin=0, xmax=1.5, 
    ymin=0.79, ymax=0.81,
    scale=.6, 
    transform shape,
    axis line style={draw=none},
    tick style={draw=none},
    xticklabels={,,},
    yticklabels={,,},
    label style={draw=none},
    height=0.2\textwidth,
    width=0.85\textwidth,
    clip=false,
]
\addplot[
    domain=0:0.07, 
    thick,
]{0.8};
\node[anchor=west] at (0.07,0.8) {\scriptsize $\frac{4\alpha + 1}{1-q_3} + \frac{3-4\alpha}{1-q_1} = \frac{8(4\alpha+3)}{3}$};
\addplot[
    domain=1.05:1.12, 
    thick,
    densely dashed,
]{0.8};
\node[anchor=west] at (1.12,0.8) {\scriptsize $6q_3 + 2q_1=5$};
\addplot[
    domain=1.7:1.77, 
    thick,
    densely dotted,
]{0.8};
\node[anchor=west] at (1.77,0.8) {\scriptsize $(12\alpha+3)q_3+(3-4\alpha)q_1=8\alpha+3$};
\end{axis}
\end{tikzpicture}
\end{subfigure}
\caption{Treatment Algorithms that Satisfy \Cref{prop:main_result}}
\label{fig:prop1}
\end{figure}
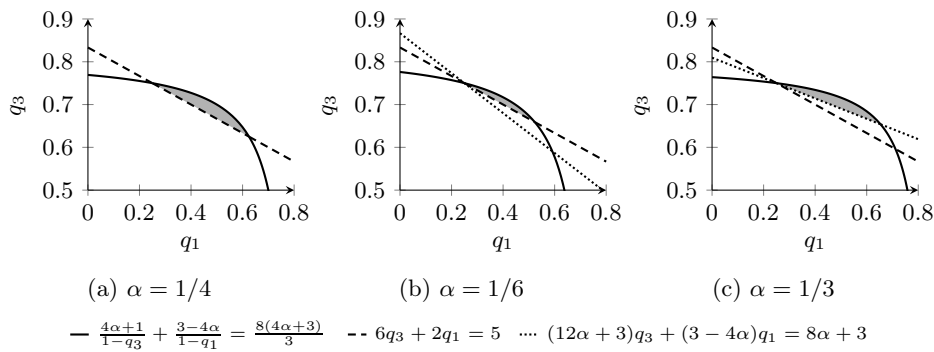

Why do $q_1$ and $q_3$ follow a concave curve to keep steady-state metrics constant, and why do they move along straight lines to maintain experimentation metrics neutral? Mathematically, the impacts of $\Delta q_1$ and $\Delta q_3$ enter the short-term experimentation metrics \textit{linearly} because only one period is measured. However, the system is generally far from converged. 
From the second period onward, each segment's state lies between the original and the new steady state, causing the effects to diminish gradually. In fact, all subsequent periods contribute to the long term calculation \textit{geometrically}. Since the horizon of the experimentation is finite, these cumulatively significant long-term effects are therefore overlooked.

As depicted in \cref{fig:prop1}, for \Cref{prop:main_result} to hold, the recommendation quality in the high segment only decreases mildly, while the low segment’s recommendation quality improves significantly.
In fact, the $q_3$ elasticity of $q_1$, formally defined as $-(\Delta q_1 / q_1^*) / (\Delta q_3 / q_3^*)$, is greater than 9 for all $q_1$ and $q_3$ that satisfy \Cref{prop:main_result}. Therefore, the intuition behind \Cref{lemma:churn-down-mass-down} generally applies: the more substantial improvement in the low segment has a more immediate impact, while the negative impact in the high segment materializes more slowly.

\section{Research Implications and Limitations}

\subsubsection{Implications} This discovery has profound and disturbing implications in the real world. First, it questions the validity of the A/B testing paradigm in RS evaluation.
A/B testing is often regarded the gold standard for determining the efficacy of RS \cite{Fabijan2023A/BFeature}. Numerous authors and RS practitioners assert the superiority of their approaches by presenting results from side-by-side tests \cite{Soria2020UsingPersonalize,Vas2021BehindSystems}. However, increases in recommendation quality metrics and retaining more users in A/B tests may merely indicate the system’s bleak long-term successfulness.

Secondly, well-intentioned developers may inadvertently harm the system. For instance, a developer might aim to enhance the low segment's performance. However, improving the recommendation quality for one segment is not always free. For a new algorithm that significantly improves the low segment but causes a seemingly minor loss in the high segment, the algorithm might be deemed an improvement in an A/B test, despite its long-term detrimental impact. On the contrary, a nefarious developer may undermine the system by cooking up a series of sabotages under the guise of improvements. 

Finally, to accurately assess the long term impacts of RS improvements, it is imperative for researchers to come up with a sophisticated experimental methodology in future studies. Questions arise such as: Should an experimentalist explicitly segment users? Is it necessary to measure system convergence? Most crucially, how can one forecast the long-term steady-state behavior from a sample with a finite horizon?

\subsubsection{Limitations}

Like all research, this piece is not without limitations. First, the findings are theoretical constructs from an idealized model, whereas real-world RSs can be affected by various noises, making it challenging to isolate retention induced biases from confounding factors. For empirical evidence that may support these theoretical findings, one may retroactively examine some natural experiments. In tech companies, restructurings due to exogenous reasons (e.g., leadership changes, layoffs, mergers) occasionally occur, disrupting RS teams. In such cases, the RS may only receive minimal maintenance, ceasing active development and naturally progressing towards its long-term state. If one can access data from these unfortunate cases, one may study the RS's real-world trajectory as it converges to its steady state.

The mechanism described here relies on two key assumptions: (1) an open system with free user inflow and churn, and (2) user churn stochastically related to recommendation quality. Thus, it doesn't apply to closed systems with fixed membership, like internal recommendation engines, or monopolistic platforms with inelastic demand. Additionally, the analysis pertains mainly to mature systems, as nascent systems in hyper-growth stages are too unstable to discuss steady states meaningfully.

The success of RS in this study is measured by the number of users. Even though it is not very controversy to assume that today's tech companies are interested in getting more users and grow, there may be exceptions. For example, because not all users generate the same amount of revenue, a revenue-maximizing stakeholder might forgo unprofitable users. This paper cannot analyze all conceivable success metrics. Nevertheless, the core principle that uneven churn leads to differences between short-term and long-term user distributions remains applicable to metrics related to users. Interested readers can use the framework presented here to analyze their own metrics.

Moreover, many real world RSs are constantly changing and some might never reach their steady states if changes are too frequent. Therefore, even though the findings here have theoretical importance on their own, one may question their practicability. For instance, can a developer constantly introduce short-term boosts (that are harmful in the long run) to indefinitely stall the downturn? This question merits further research.

\section{Conclusion}

This study examines a conceptual model of an RS, its developer, and its heterogeneous users. With the dynamics of user inflow and churn, the user distribution may achieve a steady state.
When a new recommendation algorithm is introduced, the system moves away from its current steady state and begins transitioning to a new one, which can be a long process. However, the A/B test that validates the new algorithm is usually conducted shortly after the algorithm’s development. During the transition period, the system's behavior can differ significantly from its steady state. Consequently, the metrics from A/B testing may provide misleading indicators of the new algorithm’s long term success. This discrepancy is the live testing bias caused by user retention dynamics.
In addition, retention dynamics can also introduce data bias for model training. Given the theoretical focus of this paper, examining data bias is reserved for future work.


%
%
%
\bibliographystyle{splncs04}
\bibliography{references1}

\appendix
\section*{Appendix}
\small
\begin{lemma}\label{lemmaA1}
The following results hold:
\begin{enumerate}
    \item For all $\alpha$, $q_1=1/4, q_3=3/4$ satisfy $(12\alpha+3)q_3 + (3-4\alpha)q_1 = 8\alpha+3$ and $6q_3 + 2q_1 = 5$ simultaneously.
    \item When $\alpha = 1/4$, $(12\alpha+3)q_3 + (3-4\alpha)q_1 = 8\alpha+3$ and $6q_3 + 2q_1 = 5$ are equivalent.
    \item \Cref{eq:churn-down} implies \cref{eq:q-up} if $\alpha < 1/4 \leq q_1$ or $q_1 \leq 1/4 < \alpha$.
    \item \Cref{eq:q-up} implies \cref{eq:churn-down} if $\alpha < 1/4, q_1 < 1/4$ or $\alpha > 1/4, q_1 > 1/4$.
\end{enumerate}
\end{lemma}
\begin{proof}
The first result can be obtained by plugging $q_1=1/4, q_3=3/4$ into $(12\alpha+3)q_3 + (3-4\alpha)q_1 = 8\alpha+3$ and $6q_3 + 2q_1 = 5$.
The second result can be obtained by substituting $\alpha$ with $1/4$ in $(12\alpha+3)q_3 + (3-4\alpha)q_1 = 8\alpha+3$.

To prove the third result, notice $(12\alpha+3)q_3 + (3-4\alpha)q_1 = (12\alpha+3)q_3 + (4\alpha+1)q_1 + (2-8\alpha)q_1$. By \cref{eq:churn-down}, $(12\alpha+3)q_3 + (4\alpha+1)q_1 > 5(4\alpha+1)/2$. Therefore, $(12\alpha+3)q_3 + (4\alpha+1)q_1 + (2-8\alpha)q_1 > (2-8\alpha)q_1 + 5(4\alpha+1)/2$. When $\alpha < 1/4 \leq q_1$ or $q_1 \leq 1/4 < \alpha$, $(2-8\alpha)q_1 \geq (2-8\alpha)/4$. In sum, we have $(12\alpha+3)q_3 + (3-4\alpha)q_1 > (2-8\alpha)/4 + 5(4\alpha+1)/2 = 8\alpha + 3$.

To show \cref{eq:churn-down}, it is sufficient to show $(12\alpha+3)q_3 + (4\alpha+1)q_1 > 5(4\alpha+1)/2$. Notice $(12\alpha+3)q_3 + (4\alpha+1)q_1 = (12\alpha+3)q_3 + (3-4\alpha)q_1 + (8\alpha-2)q_1$. By \cref{eq:q-up}, $(12\alpha+3)q_3 + (3-4\alpha)q_1 > 8\alpha+3$. When $\alpha < 1/4$ and $q_1 < 1/4$ or $\alpha > 1/4$ and $q_1 > 1/4$, $(8\alpha-2)q_1 > (8\alpha-2)/4$. Thus, $(12\alpha+3)q_3 + (3-4\alpha)q_1 + (8\alpha-2)q_1 > 10\alpha + 5/2$. \qed
\end{proof}

\begin{lemma}\label{lemmaA2}
If either \cref{eq:q-up,eq:mass-down} or \cref{eq:churn-down,eq:mass-down} hold, we have $q_1 > 1/4$.
\end{lemma}
\begin{proof}
For the first part, for a proof of contradiction, I assume that \cref{eq:q-up,eq:mass-down} hold for some $q_1 \leq 1/4$. Notice \cref{eq:mass-down} can be rearranged as $q_3 < 1 - 3(4\alpha+1)(1-q_1)/(8(4\alpha+3)(1-q_1) - 3(3-4\alpha))$; \cref{eq:q-up} can be rearranged as $q_3 > ((8\alpha+3)-(3-4\alpha)q_1)/(12\alpha+3)$.
For \cref{eq:q-up,eq:mass-down} to hold simultaneously, we must have $((8\alpha+3)-(3-4\alpha)q_1)/(12\alpha+3) < 1 - 3(4\alpha+1)(1-q_1)/(8(4\alpha+3)(1-q_1) - 3(3-4\alpha))$, which can be rearranged as $((3(4\alpha+3))^2(1-q_1) - (3-4\alpha)(44\alpha+15-(32\alpha+34)q_1)q_1)/((44\alpha+15) - (32\alpha+24)q_1) < 4\alpha$.
Because $q_1 \leq 1/4$, $(44\alpha+15) - (32\alpha+24)q_1 > 0$. Thus, we only require $8(3+4\alpha)q_1^2 - (40\alpha + 18)q_1 + 8\alpha + 3 < 0$.
When $q_1 \leq 1/4$, its left hand side is monotonically decreasing with respect to $q_1$ and its minimal value is obtained at $q_1 = 1/4$, which is $0$. Therefore, $8(3+4\alpha)q_1^2 - (40\alpha + 18)q_1 + 8\alpha + 3 \geq 0$ for all $q_1 \leq 1/4$, a contradiction.

For the second part, similarly, I assume that \cref{eq:churn-down,eq:mass-down} hold for some $q_1 \leq 1/4$. For \cref{eq:churn-down,eq:mass-down} to both hold, we must have $(5-2q_1)/6 < 1 - 3(4\alpha+1)(1-q_1)/(8(4\alpha+3)(1-q_1) - 3(3-4\alpha))$, which can be simplified to $(64\alpha+48)q_1^2 - (128\alpha + 24)q_1 + 28\alpha + 3 < 0$. Following the same strategy above, a similar contradiction can be constructed.\qed
\end{proof}

\subsubsection{Proof of \Cref{prop:long-term-no-conflict}}
Setting $q^T_\infty$ less (greater) than $\overline{q}^*$ yields $(4\alpha+1)((8\alpha+6)q_3 - (8\alpha+3)) / (1-q_3) < (>) (3-4\alpha)((8\alpha+3) - (8\alpha+6)q_1) / (1-q_1)$.
Because $q_1,q_3 \in (0,1)$, it can be further simplified to
\begin{equation}\label{eq:steady-state-recommendation-quality-cond}
    -32\alpha q_1 q_3 + 20\alpha q_1 + 44\alpha q_3 - 24q_1q_3 -32\alpha + 21 q_1 +15q_3 - 12 < (>) 0.
\end{equation}
Similarly, setting $m^T_\infty = 2(4\alpha+1)/(3(1-q_3)) + 2(3-4\alpha) / (3(1-q_1))$ less (greater) than $m^*$ yields $(4\alpha+1) / (1-q_3) + (3-4\alpha) / (1-q_1) < (>) 8(4\alpha+3) / 3$. Because $q_1,q_3 \in (0,1)$, it can be further simplified to
\begin{equation}\label{eq:steady-state-user-population-cond}
    -32 \alpha q_1 q_3 / 3 + 20 \alpha q_1 / 3 + 44\alpha q_3 / 3 - 8q_1q_3 - 32 \alpha / 3 + 7 q_1 + 5q_3 - 4 < (>) 0.
\end{equation}
Notice \cref{eq:steady-state-recommendation-quality-cond,eq:steady-state-user-population-cond} are identical except for a positive constant factor.\qed

\subsubsection{Proof of \Cref{lemma:q-churn-not-same}}

Plugging $m^*(x, e)$ into \cref{eq:rq_def} yields $((12\alpha + 3)q_3 + (3-4\alpha)q_1)/(8\alpha + 6)$. Setting it greater than $(8\alpha +3)/(8\alpha +6)$ results \cref{eq:q-up}. Similarly, plugging $m^*(x, e)$ into \cref{eq:cr_def} and setting it greater than $1/m^*$ results $6q_3+2q_1<5$.
Combining it with \cref{eq:q-up} results \cref{eq:lemma1}. For \cref{eq:lemma1} to hold, a necessary condition is $(5-2q_1)/6 > ((4\alpha-3)q_1 + 8\alpha+3)/(12\alpha+3)$, which simplifies to $(4\alpha-1) > 4(4\alpha-1)q_1$. Hence, \cref{eq:lemma1} can possibly hold if and only if $q_1 < 1/4 < \alpha$ or $\alpha < 1/4 < q_1$.\qed

\subsubsection{Proof of \Cref{cor:ab-win}}
Following similar steps in the proof of \Cref{lemma:q-churn-not-same}, one may reach that \cref{eq:q-up} is the condition for $\overline{q}^T_E > \overline{q}^*$ and \cref{eq:churn-down} is the condition for $\lambda^T_E < \lambda^*$. Therefore, \cref{eq:q-up,eq:churn-down} must hold simultaneously here. 

As shown in \Cref{lemmaA1}, \cref{eq:q-up,eq:churn-down} coincide when $\alpha=1/4$, which proves the first part of \Cref{cor:ab-win}.
For the second and the third points, \Cref{lemmaA1} shows that \cref{eq:q-up} is redundant if $\alpha < 1/4 \leq q_1$ or $q_1 \leq 1/4 < \alpha$ and \cref{eq:churn-down} is redundant if $\alpha < 1/4, q_1 < 1/4$ or $\alpha > 1/4, q_1 > 1/4$. Finally, leaving out the redundant inequalities and rearranging terms yield the second and the third points.\qed

\subsubsection{Proof of \Cref{lemma:q-up-mass-down}}
As shown in previous proofs, \cref{eq:q-up} is the condition for $\overline{q}^T_E > \overline{q}^*$. Following the definition of $m_\infty$, we have $ m^T_\infty = (2(4\alpha+1))/(3(1-q_3)) + (2(3-4\alpha))/(3(1-q_1))$. Setting it less than $m^*$ yields \cref{eq:mass-down}.

To show the set
\begin{equation}\label{eq:set-higher-quality-lemma}
    \left\{
    (q_1,q_3) \left| (12\alpha+3)q_3 + (3-4\alpha)q_1>8\alpha+3, \frac{4\alpha+1}{1-q_3} + \frac{3-4\alpha}{1-q_1} < \frac{8(4\alpha+3)}{3}\right.
    \right\}
\end{equation}
is not empty, it is sufficient to find one element that belongs to it. For this, we may consider the case where $q_1 = (20\alpha+9)/(32\alpha+24)$ and $q_3 = (108\alpha+63) / (128\alpha+96)$.
Plugging them into $(12\alpha+3)q_3 + (3-4\alpha)q_1 - (8\alpha+3)$ and ${(4\alpha+1)}/{(1-q_3)} + {(3-4\alpha)}/{(1-q_1)} - {8(4\alpha+3)}/3$ yields $3(4\alpha+1)(3-4\alpha)/(32(4\alpha+3)) > 0$ and $-(16(4\alpha+1)(4\alpha+3)(3-4\alpha))/(3(4\alpha+5)(20\alpha+33)) < 0$ respectively.\qed

\subsubsection{Proof of \Cref{lemma:churn-down-mass-down}} 
As shown in previous proofs, \cref{eq:mass-down} is the condition for $ m^T_\infty < m^*$; \cref{eq:churn-down} is the condition for $\lambda^T_E < \lambda^*$. Similar to the proof of \Cref{lemma:q-up-mass-down}, we need to show the set
\begin{equation}\label{eq:set-lower-churn-lemma}
    \left\{
    (q_1,q_3) \left| 6q_3 + 2q_1>5, \frac{4\alpha+1}{1-q_3} + \frac{3-4\alpha}{1-q_1} < \frac{8(4\alpha+3)}{3}\right.
    \right\}
\end{equation}
is not empty. To show it, plugging $q_1 = (16\alpha+3)/(16\alpha+12)$ and $q_3 = 3(16\alpha^2+20\alpha+9)/(8\alpha+6)^2$ into $6q_3 + 2q_1 - 5$ and ${(4\alpha+1)}/{(1-q_3)} + {(3-4\alpha)}/{(1-q_1)} - {8(4\alpha+3)}/3$ yields $(24\alpha^2)/(4\alpha+3)^2 > 0$ and $-(64\alpha^2(4\alpha+3)^2)/(9(16\alpha^2+36\alpha+9)) < 0$.\qed

\subsubsection{Proof of \Cref{prop:main_result}} 
As shown in previous proofs, \cref{eq:q-up} is the condition for $\overline{q}^T_E > \overline{q}^*$; \cref{eq:mass-down} is the condition for $m^T_\infty < m^*$; \cref{eq:churn-down} is the condition for $\overline{\lambda}^T_E < \overline{\lambda}^*$. In order for \Cref{prop:main_result} to hold, we need all \cref{eq:q-up,eq:mass-down,eq:churn-down} to hold simultaneously.

By \Cref{lemmaA2}, we have $q_1 > 1/4$ when either \cref{eq:q-up,eq:mass-down} or \cref{eq:mass-down,eq:churn-down} hold.
By \Cref{lemmaA1}, \cref{eq:q-up} implies \cref{eq:churn-down} under $\alpha > 1/4$ and $q_1 > 1/4$. Therefore, $\alpha > 1/4$ and \cref{eq:q-up,eq:mass-down} are sufficient for \cref{eq:q-up,eq:mass-down,eq:churn-down} to hold.
Similarly, \cref{eq:churn-down} implies \cref{eq:q-up} under $\alpha < 1/4$ and $q_1 > 1/4$. Therefore, $\alpha < 1/4$ and \cref{eq:mass-down,eq:churn-down} are sufficient for \cref{eq:q-up,eq:mass-down,eq:churn-down} to hold.
Lastly, \cref{eq:q-up,eq:churn-down} coincide if $\alpha = 1/4$.\qed

\end{document}